\documentclass[smallextended]{svjour3}
\smartqed{}
\usepackage{graphicx}

\usepackage{lineno,hyperref}

\usepackage{cite}
\usepackage{amssymb,amsopn,tensor,amsmath}

\newcommand{\oldcomma}{,}
\catcode`\,=13
\newcommand{,}{%
\ifmmode%
  \oldcomma\allowbreak%
\else%
  \oldcomma%
\fi%
}

\modulolinenumbers[5]

\DeclareMathOperator*{\argmin}{arg\,min}
\DeclareMathOperator*{\E}{E}
\DeclareMathOperator*{\var}{Var}
\DeclareMathOperator*{\normal}{\Phi}

\DeclareMathOperator*{\Prop}{P}
\DeclareMathOperator*{\oo}{o}
\DeclareMathOperator*{\sign}{sign}
\DeclareMathOperator{\expon}{e}
\DeclareMathOperator{\f}{f}

\begin{document}

\title{Convex method for selection of fixed effects in high-dimensional linear mixed models}

\author{Jozef Jakubík}
\institute{Slovak Academy of Sciences \\
 Institute of Measurement Science \\
 Dúbravská cesta 9 \\
 841 04 Bratislava 4 \\
 Slovakia\\
 \email{\href{mailto:jozef.jakubik.jefo@gmail.com}{jozef.jakubik.jefo@gmail.com}}  }
\date{Received: date / Accepted: date}

\maketitle

\begin{abstract}
 Analysis of high-dimensional data is currently a popular field of research, thanks to many applications e.g.\ in genetics (DNA data in genome-wide association studies), spectrometry or web analysis. At the same time, the type of problems that tend to arise in genetics can often be modelled using linear mixed models in conjunction with high-dimensional data because linear mixed models allow us to specify the covariance structure of the models. This enables us to capture relationships in data such as the population structure, family relatedness, etc.

 In this paper we introduce two new convex methods for variable selection in high-dimensional linear mixed models which, thanks to convexity, can handle many more variables than existing non-convex methods. Both methods are compared with existing methods and in the end we suggest an approach for a wider class of linear mixed models.
\end{abstract}

\section{Introduction}
The work presented in the manuscript falls into the field of model selection for linear mixed models. The field has grown extremely rapidly in the last 5--10 years, see e.g.\ the review in~\cite{muller2013model}. However, the high dimensional setting presents specific theoretical as well as computational challenges. For high-dimensional linear mixed models (LMM)~\cite{lmm}, there exist a few approaches based on \(\ell_1\) penalization. Both methods from~\cite{schell,rohart} lead in general to non-convex problems. Moreover, neither implementation (Section~\ref{sec:ss}) is effective for solving high-dimensional problems with more than \(10^4\) variables.

Frequently, the aim of data analysis with LMM is to estimate the model covariance structure, in particular the influencing variance-covariance components, but in the high dimensionality settings, variable selection from the fixed effects design matrix, say  \(\boldsymbol{X}\), followed by  parameter estimation are in the spotlight. In the case of LMM, parameters can be estimated by solving Henderson's mixed model equations~\cite{hmme,witkovsky2001matlab} or any other parameter estimation method. Traditionally, after parameter estimation of the candidate models, we can use an adequate information criterion (as e.g.,  AIC, BIC, cAIC~\cite{vaida2005}, \dots) or cross-validation to select the right model.

In this paper we focus on the selection of variables (specifically regressors from the matrix \(\boldsymbol{X}\)). We introduce convex methods for variable selection in high-dimensional linear mixed models and we prove variable selection consistency for our method.
We shall argue that if \(q<n\) (see the model specifications below), then for the purpose of variable selection it might be sufficient (and efficient) to treat LMM as a classical linear regression model, however with appropriate rescaling (weighting) of the parameters representing different parts of the random effects.
Alternatively, we suggest to consider also a more complex method (\ref{eq:W}), based on utilizing the fixed weighing matrix derived from the covariance structure of the LMM\@. However, as indicated by our simulation experiments, its positive effect in proper variable selection is only minor if compared with the more simple and computationally more effective method (\ref{eq:3}), especially if the required variance-covariance components used to derive the weighing matrix are totally unknown and should be estimated from the given data.
The considered approaches are similar to some other methods for variable selection, see e.g., the elastic net~\cite{zou2005regularization} or the adaptive LASSO~\cite{zou2006adaptive}.

In Section~\ref{sec:int} we introduce a new convex method for regressor selection. Then in Section~\ref{sec:sc} we show that the method has the theoretical property that ensures consistent variable selection with a growing number of observations. In Section~\ref{sec:vah} we propose simple and efficient methods for constructing weights. In Section~\ref{sec:ss} we compare the new methods with known methods by simulations under different scenarios.
In Section~\ref{sec:velq} we propose a generalization of the proposed methods to the case when \(q>n\) and finally, the paper is concluded Section~\ref{sec:end}.

\section{Variable selection}\label{sec:int}

LASSO (Least Absolute Shrinkage and Selection Operator)~\cite{lasso} is a popular method for parameter estimation which can be used for variable selection in the field of high-dimensional linear regression models based on \(\ell_1\) penalisation. The popularity of LASSO is due to its simplicity --- it is both easy to understand and relatively easy to compute. LASSO can be formulated as a convex problem. Thanks to progress in convex programming, LASSO problems in high-dimensional linear regression models with up to a million variables can be solved effectively.

We consider a LMM in the form~\cite{laird1982random}:
\[  \boldsymbol{\mathit{Y}}  =  \boldsymbol{X \beta}  +   \boldsymbol{Z\mathit{u}}  + \boldsymbol{ \varepsilon}, \]
where
{
 \begin{description}
  \leftskip=1cm\item[\( \boldsymbol{\mathit{Y}}\)] is \(n \times 1\) vector of observations,
  \leftskip=1cm\item [\(\boldsymbol{X}\)]	 is \(n \times p\) matrix of regressors (fixed variables),
  \leftskip=1cm\item [\(\boldsymbol{\beta}\)]	is \(p \times 1\) vector of unknown fixed effects,
  \leftskip=1cm\item[\(\boldsymbol{Z}\)] is \( n\times q\) matrix of predictors (random variables),
  \leftskip=1cm\item [\(\boldsymbol{\mathit{u}}\)]	is \(q \times 1\) vector of random effects with the distribution \(\mathcal{N}(0,\boldsymbol{D}(\boldsymbol{\theta}) )\), where \(\boldsymbol{\theta}\) represents the vector of the variance-covariance components,
  \leftskip=1cm\item [\(\boldsymbol{ \varepsilon}\)] is \(n \times 1\) error vector with the distribution \(\mathcal{N}(0, \boldsymbol{R} = \sigma^2 \boldsymbol{I})\) and independent from \(\boldsymbol{\mathit{u}}\).
 \end{description} }

We assume that only the matrix \(\boldsymbol{X}\) is high-dimensional (i.e. \(p>n\)). We shall assume that \(\boldsymbol{Z}\) is such that \(q<n\), however, in Section~\ref{sec:velq} we shall discuss in details also the case with \(q > n\). Only a small group of variables from the matrix \(\boldsymbol{X}\)  (denote it \(S^0\), and \(s^0=|S^0|\) the number of relevant variables) influence the observations \(\boldsymbol{\mathit{Y}}\). All variables from \(\boldsymbol{Z}\) are relevant in the model, but some with only a small effect (because effects are from \(\mathcal{N}(0,\boldsymbol{D})\)).\\

The structure of matrix \(\boldsymbol{D}\) may vary depending on the relationship that it captures. For LMM, the following holds:

\begin{align*}
 \E (\boldsymbol{\mathit{Y}})   & = \boldsymbol{ X\beta},                                                                 \\
 \var (\boldsymbol{\mathit{Y}}) & = \boldsymbol{ZDZ}^\mathsf{T} +\boldsymbol{R} = \boldsymbol{V}(\boldsymbol{\vartheta}),
\end{align*}
where \(\boldsymbol{\vartheta}= (\boldsymbol{\theta},\sigma^2)\) is a vector parameter of the variance-covariance components.\\

All of the mentioned methods are primarily \(\boldsymbol{\beta}\) estimation methods, not selection methods. However, they can be thought of as selection methods if we define selected variables to be those for which \(\boldsymbol{\beta}_i \neq 0\), \(i=1, \dots , p\).

The first suggested approach for variable selection in high-dimensional LMM, called HDLMMnaive, consists in a transformation that removes group effects from data. The principle of this transformation is widely used in data analysis, for example in restricted/residual maximum likelihood (REML)~\cite{reml1,reml2}. In our case we transform the data as follows
\[ \tilde{\boldsymbol{X}}  = (\boldsymbol{I} - \boldsymbol{ZZ}^+)\boldsymbol{X}, \]
\[ \tilde{\boldsymbol{\mathit{Y}}}  =   (\boldsymbol{I} - \boldsymbol{ZZ}^+)\boldsymbol{\mathit{Y}}, \]
where \(\boldsymbol{Z}^+\) is the pseudoinverse matrix of \(\boldsymbol{Z}\).
The transformation eliminates random segments of the problem (associated with the matrix \(\boldsymbol{Z}\)), which allows us to use the LASSO method for linear regression models (with dependent errors).

This is a very naive approach but as we will see in Section~\ref{sec:ss}, the transformation works well in cases when the number of variables \(q\) in the matrix \(\boldsymbol{Z}\) is relatively small relative to the number of observations \(n\),

\[q \ll n \ll p.\]

Thanks to the fast algorithm from~\cite{friedman2010regularization}, this approach can efficiently handle problems of up to \(10^6\) variables (using the \texttt{MATLAB} function \texttt{lasso}). \\

In the second suggested approach,  called LMMconvexLASSO, in contrast to  approaches in~\cite{schell,rohart} we do not penalise the log-likelihood which generally leads to a non-convex objective function. To ensure convexity, we regard the LMM a classical (i.e.\ fixed effects only) linear regression model and adapt the LASSO objective function for this purpose. Since we do not want to select (just shrink) the variables in \(\boldsymbol{Z}\), we penalise the effects \(\boldsymbol{\mathit{u}}\) using the \(\ell_2\)-norm instead of the \(\ell_1\)-norm used with the fixed effects \(\boldsymbol{\beta}\). The two penalty terms have separate penalisation parameters:
\begin{equation}
 (\hat{\boldsymbol{\beta}},\hat{\boldsymbol{\mathit{u}}})=\argmin_{\boldsymbol{\beta},\boldsymbol{\mathit{u}}}\left[  \| \boldsymbol{\mathit{Y}}-\boldsymbol{X\beta} - \boldsymbol{Z\mathit{u}}  \|_2^2 + \lambda \| \boldsymbol{\beta} \|_1 + \Lambda  \|   {\boldsymbol{\mathit{u}}}{}  \|_2^2   \right],
 \label{eq:1}
\end{equation}
where \(\lambda \) and \(\Lambda \) are the independent penalisation parameters.

Thus we obtain a two-parameter, convex problem, as opposed to the one-parameter, non-convex methods from~\cite{schell,rohart}. The extra parameter is the price we pay for convexity. A problem arises, when the vector \(\boldsymbol{\mathit{u}}\) consists of (random) effects from different groups with significantly different sizes of effects. This occurs especially in cases when the vector \(\boldsymbol{\mathit{u}}\) consists of normal distributions with significantly different variance components. In such situations, we suggest to penalise the parts of the vector \(\boldsymbol{\mathit{u}}\) with different penalisation parameters. In particular, we suggest to consider
\begin{equation}
 (\hat{\boldsymbol{\beta}},\hat{\boldsymbol{\mathit{u}}})=\argmin_{\boldsymbol{\beta},\boldsymbol{\mathit{u}}}\left[  \| \boldsymbol{\mathit{Y}}-\boldsymbol{X\beta} - \boldsymbol{Z\mathit{u}}  \|_2^2 + \lambda_0 \| \boldsymbol{\beta} \|_1 +  \sum_{i=1}^{q^*} \lambda_i  \|   \tensor[_i]{\boldsymbol{\mathit{u}}}{}  \|_2^2   \right],
 \label{eq:2}
\end{equation}
where \(q^* \) is the number of variance components without \(\sigma^2\), and  \(\tensor[_i]{\boldsymbol{\mathit{u}}}{}\) is a subvector of vector \(\boldsymbol{\mathit{u}}\) which belongs to the \(i\)-th variance component.
\(\lambda_0\) and \(\lambda_i : i \in \{ 1, \dots, q^* \} \) are parameters fixed for every minimisation.

The larger number of penalization parameters leads, however,  to computationally more complex problems. So, as a compromise method, between (\ref{eq:1}) and (\ref{eq:2}), we suggest
\begin{equation}
 (\hat{\boldsymbol{\beta}},\hat{\boldsymbol{\mathit{u}}})=\argmin_{\boldsymbol{\beta},\boldsymbol{\mathit{u}}}\left[  \| \boldsymbol{\mathit{Y}}-\boldsymbol{X\beta} - \boldsymbol{Z\mathit{u}}  \|_2^2 + \lambda \| \boldsymbol{\beta} \|_1 + \Lambda \sum_{i=1}^{q^*} w_i  \|   \tensor[_i]{\boldsymbol{\mathit{u}}}{}  \|_2^2   \right],
 \label{eq:3}
\end{equation}
where \(q^* \) is the number of variance components without \(\sigma^2\), and  \( \tensor[_i]{\boldsymbol{\mathit{u}}}{}\) is a subvector of vector \(\boldsymbol{\mathit{u}}\) which typically belongs to the \(i\)-th simple variance component of the LMM\@.
\(\lambda \) and \(\Lambda \) are penalisation parameters fixed for every minimisation and \(w_i\) are preselected weights. We take a closer look at the selection of weights in the next section.\\

The structure of the covariance matrix is naturally defined by the LMM\@. Frequently, one may have good prior information about the covariance matrix \(\boldsymbol{D}\) or the variance-covariance components \(\boldsymbol{\theta}\) of the random effects vector \(\boldsymbol{\mathit{u}}\).
This information (about \(\boldsymbol{D}(\boldsymbol{\theta})\) or simply  about the parameter \(\boldsymbol{\theta}\)) may be available from the previous studies, from the expert knowledge, or as a natural (simple) estimator from the available data.
It is then natural to generalize the weights depending on the covariance structure:
\begin{equation}
 (\hat{\boldsymbol{\beta}},\hat{\boldsymbol{\mathit{u}}})=\argmin_{\boldsymbol{\beta},\boldsymbol{\mathit{u}}}\left[  \| \boldsymbol{\mathit{Y}}-\boldsymbol{X\beta} - \boldsymbol{Z\mathit{u}}  \|_2^2 + \lambda \| \boldsymbol{\beta} \|_1 + \Lambda \boldsymbol{\mathit{u}}^\mathsf{T} \boldsymbol{W} \boldsymbol{\mathit{u}}    \right],
 \label{eq:W}
\end{equation}
where \(\boldsymbol{W}\) is the (fixed) matrix of weights based on the structure of the covariance matrix \(\boldsymbol{D}\). Ideally, we would like to set \(\boldsymbol{W}=\boldsymbol{D}(\boldsymbol{\theta}){}^{-1}\), but the vector \(\theta \) of variance-covariance components is in general unknown. Hence, in real applications, it must be approximated based on our prior knowledge or pre-estimated by using suitable (simple) variance-covariance estimation method in reasonably restricted LMM, like e.g. MINQUE~\cite{rao1971estimation,lamotte1973quadratic}.\\

The principal goal of the suggested methods is to select the relevant regressors from \(\boldsymbol{X}\). After such selection we get the the restricted LMM with (hopefully) all relevant regressors included,  and a standard methods for estimation/prediction of the fixed/random effects as well as the variance-covariance components can be applied in the second step and further used for statistical inference.

\section{Sign consistency}\label{sec:sc}

We show that method (\ref{eq:1}) is sign consistent, implying model selection consistency (the property will clearly hold for methods (\ref{eq:2}) and (\ref{eq:3}) as well). The theory and proof presented below draw upon the work of~\cite{zhao2006model,knight2000asymptotics} on linear regression.

\begin{definition}
 A method is called \emph{sign consistent} if there exist parameters \(\lambda^n=\f(n)\) and \(\Lambda \) such that
 \begin{equation*}
  \lim_{n\rightarrow \infty} \Prop ( \hat{\boldsymbol{\beta}}^n(\lambda^n,\Lambda)  =_s \boldsymbol{\beta}^0)=1,
 \end{equation*}
 where \(  \hat{\boldsymbol{\beta}}^n  (\lambda^n,\Lambda) =_s \boldsymbol{\beta}^0 \) means \(  \sign(\hat{\boldsymbol{\beta}}^n (\lambda^n,\Lambda)) = \sign(\boldsymbol{\beta}^0 )\).
\end{definition}

Without loss of generality, assume \(\boldsymbol{\beta}^{0} = (\boldsymbol{\beta}^{0}(1),  \boldsymbol{\beta}^{0}(2) ){}^\mathsf{T}= (\beta^0_1,  \dots, \beta^0_k, \beta^0_{k+1},  \dots, \beta^0_p){}^\mathsf{T} \), where \(\beta^0_j \neq 0\) for \(j=1,2,\dots, k\) and  \(\beta^0_j = 0\) for \( j=k+1,k+2,\dots, p \). Partition \(\boldsymbol{X}^n\) into \(\boldsymbol{X}^n(1)\) and \(\boldsymbol{X}^n(2)\), corresponding to \(\boldsymbol{\beta}^0(1)\) and \(\boldsymbol{\beta}^0(2)\) respectively. Let
\begin{equation*}
 \boldsymbol{\Sigma}^n = \frac{1}{n}\left[ \boldsymbol{X}^n , \boldsymbol{Z}^n\right] {}^\mathsf{T} \left[ \boldsymbol{X}^n , \boldsymbol{Z}^n\right]=
 \begin{pmatrix}
  \boldsymbol{\Sigma}_{1,1}^n & \boldsymbol{\Sigma}_{1,2}^n & \boldsymbol{\Sigma}_{1,3}^n \\
  \boldsymbol{\Sigma}_{2,1}^n & \boldsymbol{\Sigma}_{2,2}^n & \boldsymbol{\Sigma}_{2,3}^n \\
  \boldsymbol{\Sigma}_{3,1}^n & \boldsymbol{\Sigma}_{3,2}^n & \boldsymbol{\Sigma}_{3,3}^n
 \end{pmatrix}
\end{equation*}
\begin{lemma}\label{lemm:1}
 If the irrepresentable condition (\ref{ircon} in the proof) holds then
 \begin{equation*}
  \Prop ( \hat{\boldsymbol{\beta}}^n(\lambda^n,\Lambda)  =_s \boldsymbol{\beta}^0) \geq \Prop	(A^n \cap B^n),
 \end{equation*}
 for
 \begin{equation*}
  A^n = \left\lbrace    \left| (\boldsymbol{\Psi}^n){}^{-1} \boldsymbol{\Phi}^n \boldsymbol{\xi}^n \right|   <  n  \left| \boldsymbol{w} \right|  -    \frac{\lambda^n}{2}\left| (\boldsymbol{\Psi}^n){}^{-1} \boldsymbol{\theta} \right|  \right\rbrace
 \end{equation*}
 \begin{equation*}
  B^n = \left\lbrace \left|  ( \boldsymbol{\Delta}^n  (\boldsymbol{\Psi}^n ){}^{-1}  \boldsymbol{\Phi}^n   -  (\boldsymbol{X}^n(2)){}^\mathsf{T})\boldsymbol{\xi}^n \right| \leq \frac{\lambda^n}{2}\boldsymbol{\eta}  \right\rbrace .
 \end{equation*}
 where
 \begin{alignat*}{2}
  \boldsymbol{\Phi}^n & =
  \begin{pmatrix}
  (\boldsymbol{X}^n(1)){}^\mathsf{T}    \\
  (\boldsymbol{Z}^n){}^\mathsf{T}
  \end{pmatrix} , \hspace{1cm}
    & \boldsymbol{\Psi}^n & =
  \begin{pmatrix}
  \boldsymbol{\Sigma}_{1,1}^n & \boldsymbol{\Sigma}_{1,3}^n   \\
  \boldsymbol{\Sigma}_{3,1}^n & \boldsymbol{\Sigma}_{3,3}^n +\frac{\Lambda}{n}\boldsymbol{I}
  \end{pmatrix} , \\
  \boldsymbol{\Delta}^n & =
  \begin{pmatrix}
  \boldsymbol{\Sigma}_{2,1}^n  &  \boldsymbol{\Sigma}_{2,3}^n
  \end{pmatrix} ,
    & \boldsymbol{\theta} & =
  \begin{pmatrix}
  \sign(\boldsymbol{\beta}^0(1))   \\
  \boldsymbol{0}_{q\times 1}
  \end{pmatrix} .
 \end{alignat*}
\end{lemma}

\begin{proof}
 Let \(\boldsymbol{b}  = \boldsymbol{\beta} - \boldsymbol{\beta}^0 \). Then minimising the problem (\ref{eq:dok1}) or (\ref{eq:dok2}) is equivalent.
 \begin{align}
  \label{eq:dok1}
  (\hat{\boldsymbol{\beta}}^n,\hat{\boldsymbol{\mathit{u}}}^n) & =\argmin_{\boldsymbol{\beta},\boldsymbol{\mathit{u}}}\left[  \| \boldsymbol{\mathit{Y}}^n-\boldsymbol{X}^n \boldsymbol{\beta} - \boldsymbol{Z}^n \boldsymbol{\mathit{u}}  \|_2^2 + \lambda^n \| \boldsymbol{\beta} \|_1 + \Lambda  \|   {\boldsymbol{\mathit{u}}}  \|_2^2   \right]    \\
  \label{eq:dok2}
  (\hat{\boldsymbol{b}}^n,\hat{\boldsymbol{\mathit{u}}}^n )    & = \argmin_{\boldsymbol{b},\boldsymbol{\mathit{u}}} \left[ \| \boldsymbol{\xi}^n -\boldsymbol{X}^n\boldsymbol{b} - \boldsymbol{Z}^n \boldsymbol{\mathit{u}} \|_2^2 + \lambda^n \| \boldsymbol{b} + \boldsymbol{\beta}^0  \|_1 + \Lambda  \|   {\boldsymbol{\mathit{u}}} \|_2^2  \right]
 \end{align}
 where
 \begin{equation*}
  \boldsymbol{\xi}^n = \boldsymbol{\mathit{Y}}^n - \boldsymbol{X}^n \boldsymbol{\beta}^0 = \boldsymbol{\varepsilon}^n +\boldsymbol{Z}^n \boldsymbol{\mathit{u}}^0.
  \label{eq:norm}
 \end{equation*}
 The Karush-Kuhn-Tucker conditions for (\ref{eq:dok2}) are:
 \begin{align*}
  \frac{ \partial \| \boldsymbol{\xi}^n -\boldsymbol{X}^n\boldsymbol{b} - \boldsymbol{Z}^n \boldsymbol{\mathit{u}} \|_2^2 }{\partial b_j} |_{\boldsymbol{b}=\hat{\boldsymbol{b}}^n,\boldsymbol{\mathit{u}}=\hat{\boldsymbol{\mathit{u}}}^n}                 & = - \lambda^n \sign(\hat{b}_j^n) \hspace{0.3cm} & j & : \hat{b}_j^n \neq 0 \\
  \left| \frac{ \partial \| \boldsymbol{\xi}^n -\boldsymbol{X}^n\boldsymbol{b} - \boldsymbol{Z}^n \boldsymbol{\mathit{u}} \|_2^2 }{\partial b_j} |_{\boldsymbol{b}=\hat{\boldsymbol{b}}^n,\boldsymbol{\mathit{u}}=\hat{\boldsymbol{\mathit{u}}}^n}  \right| & \leq  \lambda^n                                 & j & : \hat{b}_j^n =0     \\
  \frac{ \partial \| \boldsymbol{\xi}^n -\boldsymbol{X}^n\boldsymbol{b} - \boldsymbol{Z}^n \boldsymbol{\mathit{u}} \|_2^2 }{\partial u_j} |_{\boldsymbol{b}=\hat{\boldsymbol{b}}^n,\boldsymbol{\mathit{u}}=\hat{\boldsymbol{\mathit{u}}}^n}                 & = - \Lambda \hat{u}_j^n                         & j & : \{1,2, \dots,q\}
  \label{eq:KKT}
 \end{align*}

 After performing the differentiation, we find that if there exist vectors \(\boldsymbol{b}^*, \boldsymbol{\mathit{u}}^*\) satisfying:
 \begin{equation}
  2   (\boldsymbol{X}^n(1)){}^\mathsf{T}\boldsymbol{X}^n(1)\boldsymbol{b}^* + 2 (\boldsymbol{X}^n(1)){}^\mathsf{T} \boldsymbol{Z}^n\boldsymbol{\mathit{u}}^* - 2( \boldsymbol{X}^n (1)){}^{\mathsf{T}}\boldsymbol{\xi}^n= -\lambda^n \sign(\boldsymbol{b}^*)  ,
  \label{eq:rov1}
 \end{equation}
 \begin{equation}
  -\lambda_n \boldsymbol{1} \leq 2 (\boldsymbol{X}^n(2)){}^\mathsf{T}\boldsymbol{X}^n(1)\boldsymbol{b}^* + 2 (\boldsymbol{X}^n(2)){}^\mathsf{T} \boldsymbol{Z}^n\boldsymbol{\mathit{u}}^* -  2( \boldsymbol{X}^n (2)){}^\mathsf{T}\boldsymbol{\xi}^n\leq \lambda^n \boldsymbol{1}   ,
  \label{eq:rov2}
 \end{equation}
 \begin{equation}
  2  (\boldsymbol{Z}^n){}^\mathsf{T}\boldsymbol{Z}^n \boldsymbol{\mathit{u}}^*  + 2( \boldsymbol{Z}^n ){}^\mathsf{T} \boldsymbol{X}^n(1)\boldsymbol{b}^* - 2( \boldsymbol{Z}^n ){}^\mathsf{T}\boldsymbol{\xi}^n =-2 \Lambda \boldsymbol{\mathit{u}}^*,
  \label{eq:rov3}
 \end{equation}
 then the vectors \(\hat{\boldsymbol{b}}^n =(\hat{\boldsymbol{b}}^n(1)=\boldsymbol{b}^*, \hat{\boldsymbol{b}}^n(2)=\boldsymbol{0}) \) (division of the vector \(\hat{\boldsymbol{b}}^n\) is equivalent to the division of the vector \(\hat{\boldsymbol{\beta}}^n\)) and \( \hat{\boldsymbol{\mathit{u}}}^n = \boldsymbol{\mathit{u}}^*\) are the solution of (\ref{eq:dok2}) and it holds that \(\hat{\boldsymbol{\beta}}^n(2)=0\).\\
 If instead of (\ref{eq:rov1}) we have
 \begin{equation}
  2n  (\boldsymbol{X}^n(1) ){}^\mathsf{T}\boldsymbol{X}^n(1)\boldsymbol{b}^* + 2 (\boldsymbol{X}^n(1) ){}^\mathsf{T} \boldsymbol{Z}^n\boldsymbol{\mathit{u}}^* - 2 ( \boldsymbol{X}^n (1)){}^\mathsf{T}\boldsymbol{\xi}^n= \lambda^n \sign(\boldsymbol{\beta}^0(1))  ,
  \label{eq:rov4}
 \end{equation}
 \begin{equation}
  |\boldsymbol{b}^*| < |\boldsymbol{\beta}^0(1)|,
  \label{eq:rov5}
 \end{equation}
 than \(\sign(\hat{\boldsymbol{\beta}}^n(1))=\sign(\boldsymbol{\beta}^0(1))\).\\

 Also we can bound \(|\boldsymbol{\mathit{u}}^*|<C\cdot\boldsymbol{1}_{q\times 1}~\refstepcounter{equation}(\theequation)\label{eq:rov6} \) by a constant, because \(\left|\boldsymbol{\mathit{u}}^*\right|\) is bounded.
 We use the following notation:
 \begin{equation*}
  \boldsymbol{\Psi}^n =
  \begin{pmatrix}
   \boldsymbol{\Sigma}_{1,1}^n & \boldsymbol{\Sigma}_{1,3}^n                                  \\
   \boldsymbol{\Sigma}_{3,1}^n & \boldsymbol{\Sigma}_{3,3}^n +\frac{\Lambda}{n}\boldsymbol{I}
  \end{pmatrix}, \hspace{0.5cm}
  \boldsymbol{\Delta}^n =
  \begin{pmatrix}
   \boldsymbol{\Sigma}_{2,1}^n & \boldsymbol{\Sigma}_{2,3}^n
  \end{pmatrix} , \hspace{0.5cm}
  \boldsymbol{\upsilon}^* =
  \begin{pmatrix}
   \boldsymbol{b}^*          \\
   \boldsymbol{\mathit{u}}^*
  \end{pmatrix} ,
 \end{equation*}
 \begin{equation*}
  \boldsymbol{\Phi}^n =
  \begin{pmatrix}
   (\boldsymbol{X}^n(1)){}^\mathsf{T} \\
   (\boldsymbol{Z}^n ){}^\mathsf{T}
  \end{pmatrix} ,\hspace{0.5cm}
  \boldsymbol{\theta} =
  \begin{pmatrix}
   \sign(\boldsymbol{\beta}^0(1)) \\
   \boldsymbol{0}
  \end{pmatrix}
 \end{equation*}
 and rewrite (\ref{eq:rov4}), (\ref{eq:rov2}) and (\ref{eq:rov3}):
 \begin{equation}
  \boldsymbol{\Psi}^n \boldsymbol{\upsilon}^*  =  \frac{1}{n}( \boldsymbol{\Phi}^n \boldsymbol{\xi}^n +\frac{\lambda^n}{2}\boldsymbol{\theta}) ,
  \label{eq:rov10}
 \end{equation}
 \begin{equation}
  -\frac{\lambda^n}{2n} \boldsymbol{1} \leq  \boldsymbol{\Delta}^n   \boldsymbol{\upsilon}^* - \frac{1}{n} (\boldsymbol{X}^n(2)){}^\mathsf{T}\boldsymbol{\xi}^n \leq \frac{\lambda^n}{2n} \boldsymbol{1}.
  \label{eq:rov11}
 \end{equation}
 If the matrix \(\boldsymbol{\Psi}^n\) is invertible, then we can express  \(  \boldsymbol{\upsilon}^* \) from~(\ref{eq:rov10})  and bound it using (\ref{eq:rov5}) and (\ref{eq:rov6}):
 \begin{equation*}
  | \boldsymbol{w} | =\left| \begin{pmatrix}
  \boldsymbol{\boldsymbol{\beta}}^0 (1)  \\
  C \cdot \boldsymbol{1}_{q\times 1}
  \end{pmatrix} \right|>
  \left| \begin{pmatrix}
  \boldsymbol{b}^*  \\
  \boldsymbol{\mathit{u}}^*
  \end{pmatrix} \right| = | \boldsymbol{\upsilon}^* |.
  \label{pod1}
 \end{equation*}
 There exists a solution to (\ref{eq:rov10}) if the solution to (\ref{eq:dok1}) is unique (almost sure is~\cite{tibshirani2013lasso}) and if the solution set in
 \begin{equation}
  \left| (\boldsymbol{\Psi}^n){}^{-1} \boldsymbol{\Phi}^n \boldsymbol{\xi}^n \right| <  n  \left| \boldsymbol{w} \right|  -    \frac{\lambda^n}{2}\left| (\boldsymbol{\Psi}^n){}^{-1} \boldsymbol{\theta} \right| ,
  \label{eq:rov14}
 \end{equation}
 for \(\boldsymbol{w}\) is not empty.
 We can substitute \(\boldsymbol{\upsilon}^*\) in (\ref{eq:rov11}):
 \begin{equation*}
  \left|   \boldsymbol{\Delta}^n  (\boldsymbol{\Psi}^n ){}^{-1} (( \boldsymbol{\Phi}^n \boldsymbol{\xi}^n +\lambda^n\boldsymbol{\theta})) -  (\boldsymbol{X}^n(2)){}^\mathsf{T}\boldsymbol{\xi}^n \right| \leq \frac{\lambda^n}{2} \boldsymbol{1},
  \label{eq:rov12}
 \end{equation*}
 and rewrite as
 \begin{equation}
  \left|  ( \boldsymbol{\Delta}^n  (\boldsymbol{\Psi}^n){}^{-1}  \boldsymbol{\Phi}^n   -  (\boldsymbol{X}^n(2)){}^\mathsf{T})\boldsymbol{\xi}^n \right| \leq \frac{\lambda^n}{2} (\boldsymbol{1} - \left| \boldsymbol{\Delta}^n  (\boldsymbol{\Psi}^n ){}^{-1}  \boldsymbol{\theta} \right|).
  \label{eq:rov13}
 \end{equation}
 \((\boldsymbol{1} - \left| \boldsymbol{\Delta}^n  (\boldsymbol{\Psi}^n ){}^{-1}  \boldsymbol{\theta} \right|)\) needs to be positive, therefore we define the \textbf{irrepresentable condition}. There exists a positive constant vector \(\boldsymbol{\eta}\) for which
 \begin{equation}
  \left| \boldsymbol{\Delta}^n  (\boldsymbol{\Psi}^n ){}^{-1}  \boldsymbol{\theta} \right|<1-\boldsymbol{\eta}.
  \label{ircon}
 \end{equation}
 And now (\ref{eq:rov14}) and (\ref{eq:rov13}) can be rewrite as \(A^n\) and \(B^n\) from Lemma.
\end{proof}

\(A^n\) implies that the signs of \(\boldsymbol{\beta}^0(1)\) are estimated correctly and together with \(B^n\) implies that \(\boldsymbol{\beta} ^0(2)\) are shrunk to zero.

Let
\begin{equation}
 \begin{aligned}
  (\boldsymbol{\Psi}^n){}^{-1} \boldsymbol{\Phi}^n \boldsymbol{\xi}^n /\sqrt{n} \rightarrow_d \mathcal{N}(0, \chi_1)                                                                      \\
  ( \boldsymbol{\Delta}^n  (\boldsymbol{\Psi}^n ){}^{-1}  \boldsymbol{\Phi}^n   -  (\boldsymbol{X}^n(2)){}^\mathsf{T})\boldsymbol{\xi}^n /\sqrt{n}  \rightarrow_d \mathcal{N}(0, \chi_2).
 \end{aligned}
 \label{podnorm}
\end{equation}
we assume that variance \(\chi_1\) and \(\chi_2\) are finite. As we can see in~\cite{knight2000asymptotics}, that holds if
\begin{equation*}
 \boldsymbol{\Sigma}^n \rightarrow \boldsymbol{\Sigma} \text{ \hspace{0.3cm} as \hspace{0.3cm}  } n \rightarrow \infty,
\end{equation*}
where \(\boldsymbol{\Sigma}\) is a positive definite matrix. And,
\begin{equation*}
 \frac{1}{n} \max_{1\leq i \leq n} (\left[ (\boldsymbol{x}_i^n, \boldsymbol{z}_i^n ) \right] {}^\mathsf{T}\left[ \boldsymbol{x}_i^n,\boldsymbol{z}_i^n\right] ) \rightarrow 0 \text{ \hspace{0.3cm} as \hspace{0.3cm}  } n \rightarrow \infty.
\end{equation*}

\begin{theorem}\label{theorem1}
 Method (\ref{eq:1}) is sign consistent for \(\lambda_n\) satisfying \(\lambda_n / n \rightarrow 0\) and \(\lambda_n / n^\frac{1+c}{2} \rightarrow \infty \), where \(0\leq c < 1\) under the conditions of finite variance matrices \(\chi_1\), \(\chi_2\) and the irrepresentable condition, and we have:
 \[\Prop(\hat{\boldsymbol{\beta}}(\lambda_n)=_s \boldsymbol{\beta}^0 )= 1- \oo(\expon^{-n^c})\]
\end{theorem}

\begin{proof}
 By Lemma~\ref{lemm:1} if the irrepresentable condition holds
 \begin{equation*}
  \Prop ( \hat{\boldsymbol{\beta}}^n(\lambda^n,\Lambda)  =_s \boldsymbol{\beta}^0) \geq \Prop	(A^n \cap B^n).
 \end{equation*}
 It follows that
 \begin{align*}
  1- \Prop	(A^n \cap B^n) & \leq \Prop	((A^n){}^C) + \Prop ((B^n){}^C)                                                                                                                                                            \\
                          & \leq \sum_{i=1}^{k+q} \Prop	(|{\gamma}^n_i|\geq \sqrt{n}(|{\boldsymbol{w}}_i|-\frac{\lambda^n}{2n} {\tau}^n_i)) + \sum_{i=1}^{p-k} \Prop (|{\delta}^n_i| \geq \frac{\lambda^n}{2\sqrt{n}} {\eta}_i) ,
 \end{align*}
 where
 \begin{align*}
  \boldsymbol{\gamma}^n & =({\gamma}^n_1, {\gamma}^n_2, \dots, {\gamma}^n_{(k+q)}){}^\mathsf{T}  =( \boldsymbol{\Psi}^n){}^{-1} \boldsymbol{\Phi}^n \boldsymbol{\xi}^n/\sqrt{n}                                                                      \\
  \boldsymbol{\delta}^n & = ({\delta}^n_1, {\delta}^n_2, \dots, {\delta}^n_{(p-k)}){}^\mathsf{T} = ( \boldsymbol{\Delta}^n  (\boldsymbol{\Psi}^n ){}^{-1}  \boldsymbol{\Phi}^n   -  (\boldsymbol{X}^n(2)){}^\mathsf{T})\boldsymbol{\xi}^n / \sqrt{n} \\
  \boldsymbol{\tau}^n   & =  ({\tau}^n_1, {\tau}^n_2, \dots, {\tau}^n_{(k+q)}){}^\mathsf{T} = (\boldsymbol{\Psi}^n){}^{-1} \boldsymbol{\theta} .
 \end{align*}
 For \(t > 0\), the Gaussian distribution has its tail probability bounded by
 \begin{equation*}
  1-\normal(t) < t^{-1} \expon^{-\frac{1}{2}t^2}
 \end{equation*}
 therefore
 \begin{align*}
  \sum_{i=1}^{k+q}  \Prop	(|\boldsymbol{\gamma}^n_i|\geq \sqrt{n}(|{\boldsymbol{w}}_i|-\frac{\lambda^n}{2n} {\tau}^n_i) )
    & \leq  (1-\oo(1))  \sum_{i=1}^{k+q}   (1- \normal((1+\oo(1))\frac{1}{s} \sqrt{n}|{\boldsymbol{w}}_i|)	) \\
    & = \oo (\expon ^{-n^c}),
 \end{align*}
 and
 \begin{equation*}
  \sum_{i=1}^{p-k} \Prop (|\boldsymbol{\delta}^n_i| \geq \frac{\lambda^n}{2\sqrt{n}} \boldsymbol{\eta}_i)
  = \sum_{i=1}^{p-k}  (1- \normal(\frac{\lambda^n}{2s\sqrt{n}}  \boldsymbol{\eta}_i )	)
  = \oo (\expon ^{-n^c}).
 \end{equation*}
 Theorem~\ref{theorem1} follows immediately.
\end{proof}

\section{Selection of weights}\label{sec:vah}

Investigating all combinations of penalisation parameters in the case of equation (\ref{eq:2}) can be very time consuming, because the number of parameter combinations grows exponentially with the number of variance components. However, in many cases it is not necessary to investigate all combinations and it suffices to replace all parameters by one, as in the case of equation (\ref{eq:1}). Below, in Section~\ref{sec:ss}, Figure~\ref{fig:obr3}, we can see that the replacement of all parameters \(\lambda_i : i \in \{ 1, \dots, q^* \} \) by one parameter \(\Lambda \) gives in many cases identical or very similar results as investigating all combinations of lambdas. Assuming that the computing time of a single optimisation problem is roughly the same for each parameter combination, in the case of equation (\ref{eq:1}) we must investigate \(k^2\) combinations and in the case of equation (\ref{eq:2}) we must investigate \(k^{q^*}\) combinations, which can take considerably more time.

On the other hand, as shown in Figure (\ref{fig:obr4}), the replacing of all parameters \(\lambda_i : i \in \{ 1, \dots, q^* \} \) by one parameter \(\Lambda \) leads in some cases to worse results. It is obvious that a suitable preselection of weights \(w_i\) in the case of equation (\ref{eq:3}) can lead to exactly the same results as investigating all combinations of penalisation parameters in the case of equation (\ref{eq:2}). Moreover, solving the problem (\ref{eq:3}) can be equally fast as solving (\ref{eq:1}).

The preselection of good weights \(w_i\) is crucial. In our simulation study (Section~\ref{sec:ss}) we use the following weights:
\begin{equation}
 w_i = \frac{1-\tensor[_i]{\theta}{}}{\tensor[_i]{q}{}},
 \label{eq:4}
\end{equation}
where \(\tensor[_i]{q}{}\) is the number of variables in matrix \(\boldsymbol{Z}\) belonging to the \(i\)-th variance component (the number of effects in the subvector \(\tensor[_i]{\boldsymbol{\mathit{u}}}{}\)), \(q= \sum_{i=1}^{q^*} \tensor[_i]{q}{}\). \(\tensor[_i]{\theta}{}\) is the average absolute value of correlation between the variables from matrix \(\boldsymbol{Z}\) belonging to the \(i\)-th variance component (\(\tensor[_{\boldsymbol{\mathit{u}}_i}]{\boldsymbol{Z}}{}\)) and the observation \(\boldsymbol{\mathit{Y}}\):
\[
 \tensor[_i]{\theta}{} = \frac{\sum_{i=1}^{\tensor[_i]{q}{}}|\rho (\tensor[_{\boldsymbol{\mathit{u}}_i}]{\boldsymbol{Z}}{_{(:,i)}},\boldsymbol{\mathit{Y}})|}{\tensor[_i]{q}{}}
\]
This preselection of weights adjusts the norm of the subvectors \(\tensor[_i]{\boldsymbol{\mathit{u}}}{}\) for their dimension, and at the same time places a greater weight on effects with smaller average `effects'. \\

\section{Simulation study}\label{sec:ss}

Taking into account the aim of the study, we compare different approaches, which can be used for variable selection in the high-dimensional LMM\@.

\subsection{Compared methods}

\begin{description}
 \item[LASSO]~\cite{lasso,buhlgeer} is an established method which can be used for selecting variables in linear regression models.

 In this study we use the LASSO as the reference, as it ignores the LMM data structure (ignores random part of problem). For the LASSO method we use the built-in \texttt{MATLAB} function \texttt{lasso}.\\

 \item[LMMLASSO] from~\cite{schell} is a method based on the minimisation of the non-convex objective function consisting of the \(\ell_1\) penalised negative log-likelihood with respect to the parameter \(\boldsymbol{ \beta}\) from \(\boldsymbol{\mathit{Y}}  \sim \mathcal{N}(\boldsymbol{ X\beta},\boldsymbol{V}(=\boldsymbol{ZDZ}^\mathsf{T} +\boldsymbol{R}))\):

 \[(\hat{\boldsymbol{\beta}},\hat{\boldsymbol{D}},\hat{\boldsymbol{R}} ) =\argmin_{\boldsymbol{\beta},\boldsymbol{D},\boldsymbol{R} } \left[  \frac{1}{2} \log|\boldsymbol{V}| +\frac{1}{2} (\boldsymbol{\mathit{Y}} -\boldsymbol{ X\beta}){}^\mathsf{T}  \boldsymbol{V}^{-1}(\boldsymbol{\mathit{Y}} -\boldsymbol{ X\beta})  +\lambda ‎‎\|\boldsymbol{\beta}\|_1\right] ,\]

 where \(\lambda \) is a fixed parameter. For this method we used the language \texttt{R} package \texttt{lmmlasso}, which uses the coordinate gradient descent algorithm to optimise the objective function.\\

 \item[LASSOP] from~\cite{rohart} is a method based on the log-likelihood of the compound data \((\boldsymbol{\mathit{Y}}^\mathsf{T},\boldsymbol{\mathit{u}}^\mathsf{T}){}^\mathsf{T}\) penalised with the \(\ell_1\) penalisation:

 \begin{multline*}
  (\hat{\boldsymbol{\beta}},\hat{\boldsymbol{D}},\hat{\boldsymbol{R}} ) =\argmin_{\boldsymbol{\beta},\boldsymbol{D},\boldsymbol{R} }  \left[   \log|\boldsymbol{ R}| + (\boldsymbol{\mathit{Y}}-\boldsymbol{X\beta} - \boldsymbol{Z\mathit{u}}){}^\mathsf{T} \boldsymbol{ R}^{-1}(\boldsymbol{\mathit{Y}}-\boldsymbol{ X\beta} -\boldsymbol{Z\mathit{u}}) \right.   \\  \left. +  \log|\boldsymbol{ D}| + \boldsymbol{\mathit{u}}^\mathsf{T} \boldsymbol{ D}^{-1}\boldsymbol{\mathit{u}}  +\lambda ‎‎\|\boldsymbol{ \beta}\|_1 \right] ,
 \end{multline*}

 where \(\lambda \)  is a fixed parameter. The objective function is non-convex, like in LMMLASSO\@. This method is implemented in language \texttt{R} in package \texttt{MMS}. The optimisation problem in this implementation is solved by the adjusted EM algorithm.\\

 \item[LMM-LASSO] from~\cite{Lippert2013} is an approach for LMM with one variance component \(\sigma_D^2 \). The method is suitable for high dimensional data and it is based on a data transformation, which eliminates correlation between observations. We first estimate \(\sigma_D^2 \), \(\sigma^2\) by Maximum Likelihood under the null model, ignoring the effect of variables in matrix \(\boldsymbol{X}\). Let \(\boldsymbol{K}=1/q \cdot \boldsymbol{ZZ}^\mathsf{T} \). Having fixed \(\hat{\gamma} = \hat{\sigma_D^2} / \hat{\sigma^2}\), we use the spectral decomposition of \(\boldsymbol{K} = \boldsymbol{\mathit{u}} \boldsymbol{\Lambda} \boldsymbol{\mathit{u}}^\mathsf{T} \) to rotate our data, so that the covariance matrix becomes isotropic:
 \begin{align*}
  \tilde{\boldsymbol{X}}          & = (\hat{\gamma} \boldsymbol{\Lambda} + \boldsymbol{I} ){}^{-\frac{1}{2}} \boldsymbol{\mathit{u}}^\mathsf{T} \boldsymbol{X}              \\
  \tilde{\boldsymbol{\mathit{Y}}} & =   (\hat{\gamma} \boldsymbol{\Lambda} + \boldsymbol{I} ){}^{-\frac{1}{2}} \boldsymbol{\mathit{u}}^\mathsf{T} \boldsymbol{\mathit{Y}} .
 \end{align*}
 After transforming the data we use the LASSO method
 \[\hat{\beta}= \argmin_{\beta} \left[ \frac{1}{\hat{\sigma^2} }\| \tilde{\boldsymbol{\mathit{Y}}} - \tilde{\boldsymbol{X}} \boldsymbol{\beta} \|^2_2 + \lambda \| \boldsymbol{\beta} \|_1\right] .\]
 We implement this method in \texttt{MATLAB}.\\

 \item[HDLMMnaive \& LMMconvexLASSO] from Section~\ref{sec:int}. Both methods are implemented in \texttt{MATLAB}. For convex optimisation, we use the modelling system for convex optimisation \texttt{CVX}~\cite{cvx} with solver \texttt{Mosek}~\cite{mosek}. In the LMMconvexLASSO approach, we implement the solution to problem (\ref{eq:3}) with weights defined in (\ref{eq:4}). \\
\end{description}

\subsection{Simulation study design}

We compare our two methods step by step with other methods, because methods LMMLASSO and LASSOP  solve a different type of LMM than LMM-LASSO\@. In each comparison, we generate a hundred problems as described in the next parts.

As a correctly solved problem we consider only a problem for which the method gives exactly the set \(S^0\).
All figures show the number of correctly solved problems for all used methods for different numbers of relevant variables (from \(1\) to \(10\) or from \(1\) to \(20\)).
Unless otherwise stated, the elements of the variables (i.e.\ the columns of the design matrix \(\boldsymbol{X}\)) are independently generated from  the uniform \(\left\langle 0,1\right\rangle \) distribution and normalised.\\

First of all (Figure~\ref{fig:obr1}) we compare our methods with  LMMLASSO and LASSOP on high-dimensional data with a `small dimension', because the current implementations of methods LMMLASSO and LASSOP are usually not able to solve problems of dimension higher than \(p=10^3\).

Data in this simulation study are divided into twenty groups of six observations. Together we have \(n=120\) observations. For each observation we observe \(p=150\) variables, but only \(s^0 = \{  1, \dots , 10 \} \) variables influence the observations. Relevant variables are randomly selected from all variables and the effect of relevant variables is one. The effect of other variables is zero. The matrix  \(\boldsymbol{Z}\) captures the group structure of the data. For every group we observe two variables, therefore we consider two variance components and the error variance component. \(\boldsymbol{Z}\) is a block diagonal matrix and \(\boldsymbol{\mathit{u}}\) consist of two parts, each for one variance component. Both parts of the random effects \(\boldsymbol{\mathit{u}}\) are randomly selected from \(\mathcal{N}(0, \boldsymbol{D}  = 2 \cdot \boldsymbol{I})\). Errors are from \(\mathcal{N}(0,\boldsymbol{I})\). This example is inspired by an example from the package \texttt{lmmlasso}.\\

\begin{figure}[h]
 \centering

 \includegraphics[width=.9\textwidth]{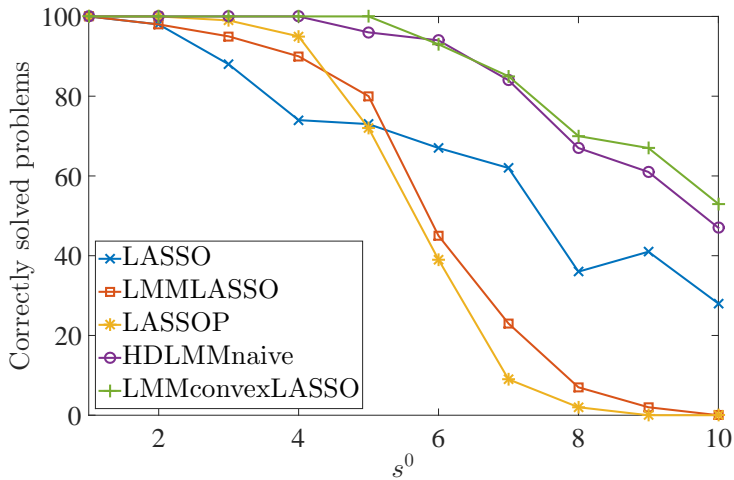}

 \caption{Comparison of ability to find exactly the set \(S^0\). The comparison is performed on high-dimensional data with only a `small dimension', because the methods LMMLASSO and LASSOP can solve only problems with dimension up to \(p = 10^3\).}\label{fig:obr1}
\end{figure}

In Figure~\ref{fig:obr1}, the LASSO is doing surprisingly well. This may be because for every observation there are just two random effects affecting vector \(\boldsymbol{\mathit{Y}}\). The rapid deterioration in the performance of LMMLASSO and LASSOP is in our opinion caused by bad implementation of the methods (Section~\ref{sec:ss}).\\

Second, we compare our methods with the method LMM-LASSO on high-dimensional data with one variance component (see Figure~\ref{fig:obr2}).

In this case we have \(n=200\) observations divided into twenty groups of ten observations. For each observation we observe \(p=5000\) variables, \(s^0 = \{  1, \dots , 20 \} \) all with effect one. The matrix \(\boldsymbol{Z}\) captures the group structure of the data. \(\boldsymbol{Z}_{i,j}\) is \(1\) if the \(i\)-th observation belongs to the \(j\)-th group and \(0\) otherwise. The random effects \(u\) are randomly selected from \(\mathcal{N}(0, \boldsymbol{I})\). Errors are from \(\mathcal{N}(0,0.2\cdot \boldsymbol{I})\). \\

\begin{figure}[h]
 \centering

 \includegraphics[width=.9\textwidth]{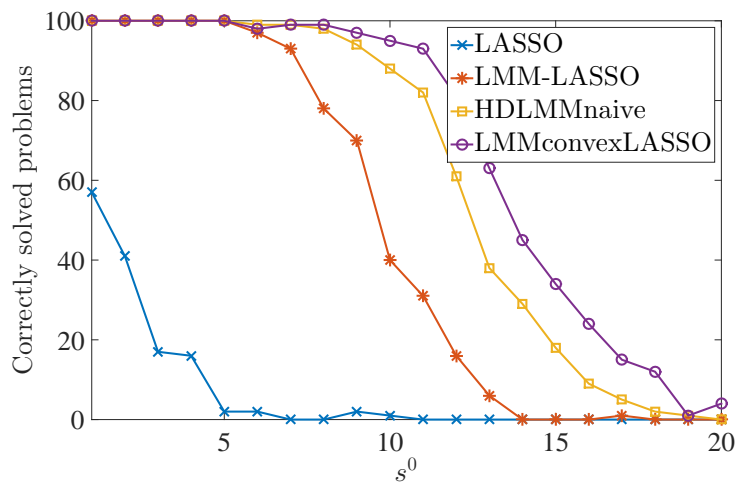}

 \caption{Comparison with the method LMM-LASSO designed for problems with one variance component (apart from the error variance component).}\label{fig:obr2}
\end{figure}

Figure~\ref{fig:obrD2} show a comparison of methods HDLMMnaive, LMMconvexLASSO (version (\ref{eq:3})) and LMMconvexLASSO (version (\ref{eq:W}) with \(\boldsymbol{W} = \boldsymbol{D}^{-1}\)) in similar scenarios but with greater correlation in matrix \(\boldsymbol{D}\).
\begin{itemize}
 \item[\(\bullet \)] In the first case \(\boldsymbol{D}\) is the identity matrix.
 \item[\(\bullet \)] In the second case it is diagonal with the first half of elements equal to \(2\) and the second half \(0.8\).
 \item[\(\bullet \)] In the third case the main diagonal is the same as in the second case, but on the first diagonal below and above the main diagonal we put \(0.9\).
 \item[\(\bullet \)] In the fourth case the main diagonal is the same as in the second case, but on the first three diagonals below and above the main diagonal we put \(0.9\), \(0.8\) and \(0.7\).
 \item[\(\bullet \)] In the fifth case the main diagonal is the same as in the second case but both block matrices are filled with \(0.8\).
\end{itemize}
Moreover we observe \(n = 200\) observations with \(p=5000\) regressors, but only \(s_0 = 10\) of them are relevant for observations \(\boldsymbol{\mathit{Y}}\). The number of predictors is \(q = 40\) and the variance of errors is \(0.2\).

\begin{figure}[h]
 \centering

 \includegraphics[width=.9\textwidth]{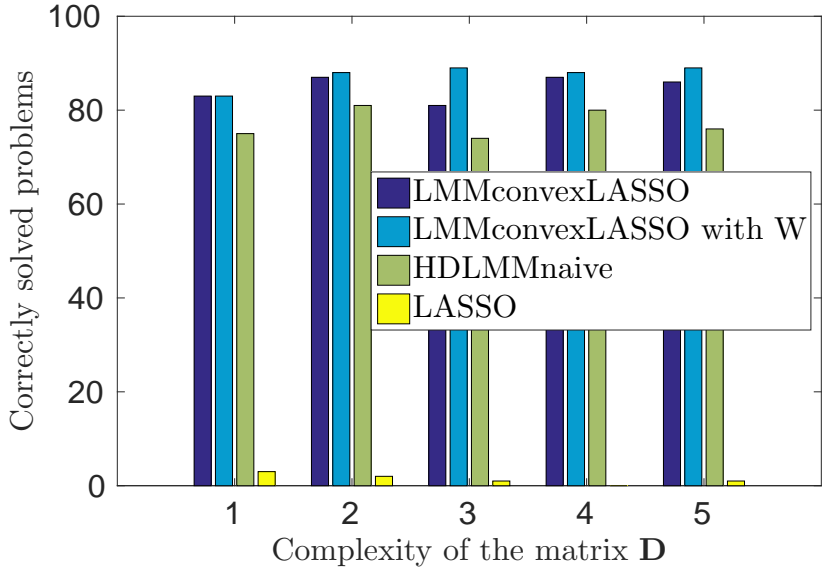}

 \caption{Comparison of LMMconvexLASSO (version (\ref{eq:3}) and version (\ref{eq:W}) with \(\boldsymbol{W}= \boldsymbol{D}^{-1}\)) and HDLMMnaive approaches in five cases of growing complexity of the matrix \(\boldsymbol{D}\), from diagonal to block diagonal matrix with dense block matrices.}\label{fig:obrD2}
\end{figure}

\newpage

Finally, Figures~\ref{fig:obr3} and~\ref{fig:obr4} show a comparison of all four our methods: HDLMMnaive and LMMconvexLASSO in three different formulations as given in eq. (\ref{eq:1}), eq. (\ref{eq:2}) and eq. (\ref{eq:3}). Two versions of data are used.
\(n=200\) observations are divided into twenty groups of ten observations. \(p=10^4\) variables, \(s^0 = \{  1, \dots , 10 \} \) variables influence the observations with effect one. \(\boldsymbol{Z}\) is a block diagonal matrix, \(\boldsymbol{D}\) is a diagonal matrix and \(\boldsymbol{\mathit{u}}\) consists of three parts, each for one variance component.\\

We consider two scenarios. In the first one, (Figure~\ref{fig:obr3}), the variance components are \(1\), \(1.2\), \(0.8\) and the error variance component is \(0.1\).

\newpage

\begin{figure}[h]
 \centering

 \includegraphics[width=.9\textwidth]{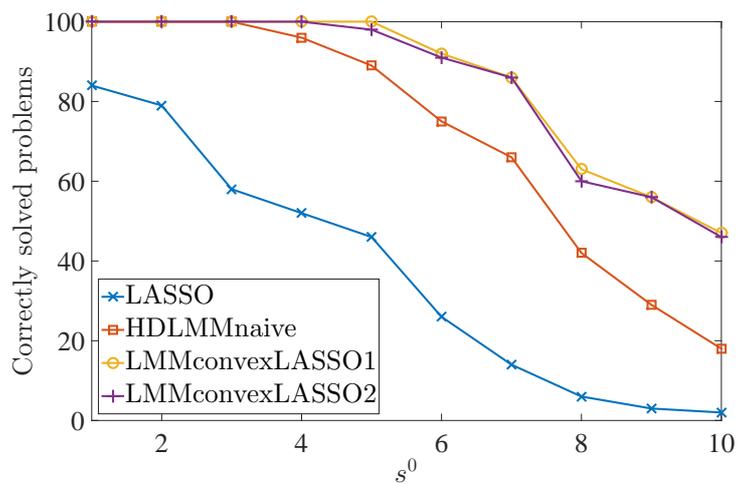}

 \caption{Comparison of different approaches to penalisation described in Section~\ref{sec:int} on data with a small difference between variance components. LMMconvexLASSO2 is the approach with one penalisation parameter and LMMconvexLASSO1 is the multi-parameter approach. In this case we omit the weighted approach because it gives the same results as the multi-parameter approach.}\label{fig:obr3}
\end{figure}

In the second scenario, (Figure~\ref{fig:obr4}), the variance components are \(2\), \(4\), \(0.5\) and the error variance component is \(0.1\).

\begin{figure}[h]
 \centering

 \includegraphics[width=.9\textwidth]{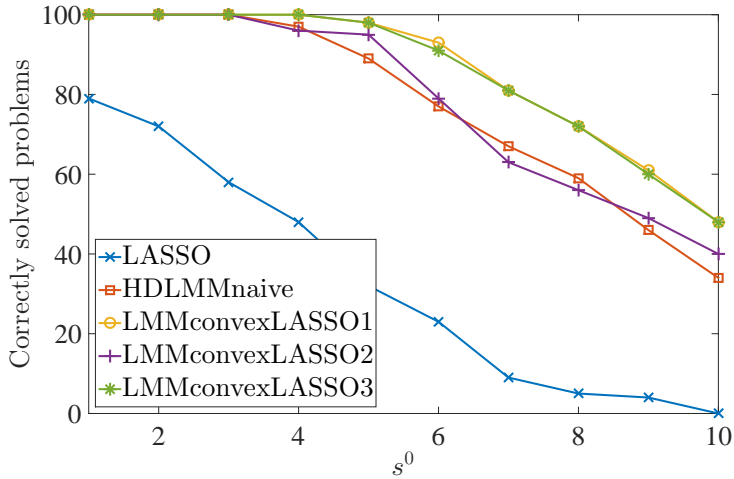}

 \caption{Comparison of different approaches to penalisation described in Section~\ref{sec:int} on data with a greater difference between variance components. LMMconvexLASSO2 is the approach with one penalisation parameter, LMMconvexLASSO1 is the multi-parameter approach and LMMconvexLASSO3 is the weighted approach.}\label{fig:obr4}
\end{figure}

\newpage

\subsection{Discussion}

The models in our simulation studies are essentially quite simple, and therefore the absolute results are very good, especially for smaller \(s^0\). However, the comparison of individual methods is more important. The success of our methods can be attributed to the fact that they do not aim to estimate the matrix \(\boldsymbol{D}\) and directly estimate the vector \(\boldsymbol{\mathit{u}}\). They avoid errors from double estimating, which arise when one first estimates the matrix \(\boldsymbol{D}\) and subsequently estimates the vector \(\boldsymbol{\mathit{u}}\) based on the estimate of \(\boldsymbol{D}\). A possible disadvantage of our methods is that they do not provide a direct estimate of matrix \(\boldsymbol{D}\). However, our main goal is dealing with the high-dimensional matrix \(\boldsymbol{X}\),  not the estimation of matrix \(\boldsymbol{D}\).

Both new methods work well. The more complex method LMMconvexLASSO performs better than HDLMMnaive in almost all cases, but HDLMMnaive is faster and it can handle more variables.

At the same time, we can notice a difference between Figures~\ref{fig:obr3} and~\ref{fig:obr4}, where it turns out that for a small difference between variance components, one penalisation parameter is enough. In contrast, if the differences between variance components are greater, then the version with more penalisation parameters performs better than the version with one penalisation parameter.  The weighted version with one penalisation parameter performs almost identically to the version with more penalisation parameters in both cases.

The other great advantage of our approach is convexity, and the possibility to use established and well working software for optimisation.

We may notice that with the growing size of the random effect vector \(\boldsymbol{u}\), our methods fail more often because more variables also mean more `freedom' in the optimisation process, which in turn implies `worse performance'. This leads us to the problem `What to do if \(q > n\)'.\\

\section{High-dimensional matrix \(\boldsymbol{Z}\) (\(q>n\))}\label{sec:velq}

Often, we encounter an LMM that has dimension \(q\) of the matrix \(\boldsymbol{Z}\) greater than the number of observations \(n\). At the time of writing, we are not aware of any approach that could handle a high-dimensional matrix \(\boldsymbol{X}\) as well as a high-dimensional matrix \(\boldsymbol{Z}\). In this section we introduce a two step approach for this type of LMMs.

We start with an example. Suppose we observe the growth of plants in different places on Earth and we investigate how the genetic information, weather and the composition of the soil influence the growth. We capture genetic information in the matrix \(\boldsymbol{X}\). A lot of genetic information can be redundant and our aim is to select the relevant genetic information.  The soil can contain various substances and we also have a long time series of daily weather. All this information can be captured in the matrix \(\boldsymbol{Z}\). In this setting, the number of random variables \(q\) is greater than the number of observed plants \(n\). We assume that the effects of random variables is normally distributed with two variance components, one for weather and one for soil. Our model is an LMM with high-dimensional matrices \(\boldsymbol{X}\) and \(\boldsymbol{Z}\).

Our approach consists of two steps. In the first step, we reduce the dimension \(q\) of the matrix \(\boldsymbol{Z}\) by creating new variables as linear combinations of the original ones. In our example, we would create new soil and weather type variables. With the smaller number of soil and weather types, we may not be able to capture all the soil and weather information exactly. In our simulation study we set the number of new variables so that they capture at least \(95\% \) the original soil and weather types. For creating the new variables, we use PCA (principal component analysis)~\cite{wold1987principal,jolliffe2002principal}, but it is possible to use any of the existing methods for dimension reduction. In all cases in our simulation study, the number of new variables was smaller than the number of observations and in the second step, we can use the LMMconvexLASSO\@. It is clear that the effect of the new variables is also normally distributed and that the use of LMM methods is legitimate.\\

\subsection{Simulation study}\label{sec:ss1}

We proceed almost like in Section~\ref{sec:ss}.  We generate matrices \(\boldsymbol{X}\) with dimensions \(n=200\), \(p=2000\). We create fifty different types of soil, each as combination of \(200\) substances and we create twenty types of weather, each as a \(200\) day long  time series holding the number of minutes that the sun was shining. The matrix \(\boldsymbol{Z}\) was created as a random combination of weather type and soil type (\(q=400\), \(200\) variables for weather and \(200\) variables for soil).  Only \(s^0 = \{  1, \dots , 10 \} \) variables from matrix \(\boldsymbol{X}\) influence the observations with effect one and both parts of the random effects \(\boldsymbol{\mathit{u}}\) are randomly selected from \(\mathcal{N}(0, \boldsymbol{D}  =  \boldsymbol{I})\). Errors are from \(\mathcal{N}(0, 0.2 \cdot \boldsymbol{I})\).

Figure~\ref{fig:obr5} compares LASSO (which ignores the LMM data structure) and the approach described in the previous section.

\begin{figure}[h]
 \centering

 \includegraphics[width=.9\textwidth]{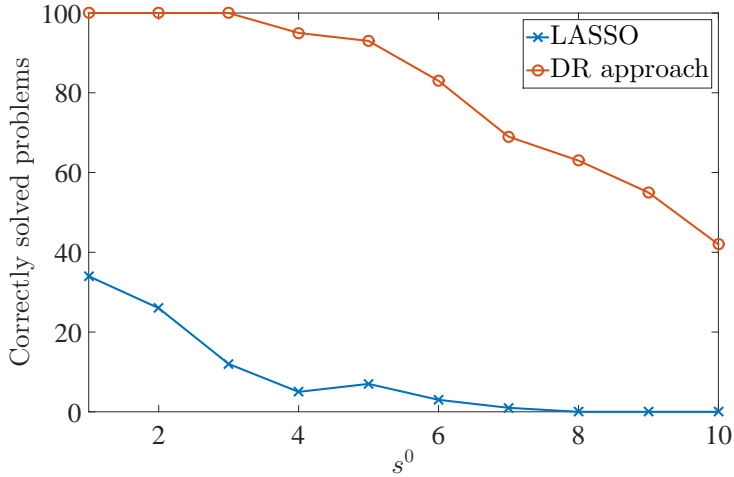}

 \caption{Comparison of LASSO (which ignores the LMM data structure) with our dimension reduction (DR) approach on data with more random variables \(q\) than observations \(n\).}\label{fig:obr5}
\end{figure}

\subsection{Discussion}

As can be seen in Figure~\ref{fig:obr5}, the proposed approach works very well in our simple simulation study, while the LASSO completely blew up. The main reason why the LASSO did not work so well was that it ignored a lot of random effect.

The average number of new variables was \(43.4\) for soil and \(16.2\) for weather. After the phase of variable selection, either the model with original variables or the model with new variables can be estimated as a standard LMM, for example via Henderson's mixed model equation. \\

\section{Conclusion}\label{sec:end}

In Section~\ref{sec:int}, we introduced two new methods for variable selection in high-dimensional LMMs, and in Section~\ref{sec:vah}, we designed a weighting which may replace the searching through all possible parameter combinations. The greatest advantage of our methods is convexity and the associated ability to handle high-dimensional data with dimension up to \(10^5\) in the case of LMMconvexLASSO\@.

In Section~\ref{sec:ss}, we compared our methods LMMconvexLASSO with weights from Section~\ref{sec:vah} and HDLMMnaive with other existing methods. As can be seen in Figures~\ref{fig:obr1} and~\ref{fig:obr2}, our methods always perform better than the other methods in our simulation study. From our comparisons, it seems that the more complex method LMMconvexLASSO performs better than the method HDLMMnaive (see Figures~\ref{fig:obr3} and~\ref{fig:obr4}).

We also show in Section~\ref{sec:sc} that the introduced methods are consistent. With a sufficient number of observations, the simpler methods (\ref{eq:1}), (\ref{eq:2}), (\ref{eq:3}) are capable of finding the true set of relevant regressors.

Section~\ref{sec:ss} also shows that if \(q<n\) it might be sufficient for the purpose of variable selection to consider LMM as classical linear regression (random vector effects \(\boldsymbol{\mathit{u}}\) are considered as fixed). Alternatively, as shown in Figure~\ref{fig:obrD2}, the rescaling of different parts of \(\boldsymbol{\mathit{u}}\) (method (\ref{eq:3})) can be sufficient. The use of the more complex method (\ref{eq:W}) has only minor positive effect with respect to the computationally more effective method (\ref{eq:3}), especially if the required variance-covariance components used to derive the weighing matrix are totally unknown and should be estimated from the given data.

Both of our methods are not suitable for in cases when the dimension \(q\) of matrix \(\boldsymbol{Z}\) is greater than the number of observations \(n\), but in Section~\ref{sec:velq} we proposed an approach for these cases. As shown in the simulation study of Section~\ref{sec:ss1}, this approach works relatively well.

\texttt{MATLAB} source codes for both of our methods are available on \url{http://www.mathworks.com/matlabcentral/fileexchange/56952-lmmconvexlasso}.\\

\begin{acknowledgements}
 The~work was supported by~the~Slovak Research and Development Agency, project APVV--15--0295, and by~the~Scientific Grant Agency VEGA of the Ministry of Education of the Slovak Republic and the Slovak Academy of~Sciences, by~the~projects VEGA 2/0047/15 and VEGA 2/0011/16.
\end{acknowledgements}

\bibliographystyle{spmpsci}      
\bibliography{ref}
\end{document}